\RequirePackage{fix-cm}
\documentclass[twocolumn]{svjour3}          
\smartqed  
\usepackage{graphicx}
\usepackage{cite}
\usepackage{amssymb,amsmath}
\usepackage{tensor}
\usepackage{subfigure}
\usepackage[ruled,vlined]{algorithm2e}
\usepackage{color}
\newcommand{\lbn}{\par\medskip\noindent}

\newcommand{\tr}{\mathrm{tr}}

\begin{document}

\title{Performance Analysis of ZF Receivers with Imperfect CSI for Uplink Massive MIMO Systems\thanks{This work was supported by Basic Research Laboratories (BRL) through NRF grant funded by the MSIP (No.~2015056354).}
}


\author{Van-Dinh Nguyen \and
        Oh-Soon Shin}


\institute{Van-Dinh Nguyen \at
              School of Electronic Engineering \\
							Soongsil University, Seoul 06978, Korea \\
              \email{nguyenvandinh@ssu.ac.kr}           
           \and
           Oh-Soon Shin \at
              School of Electronic Engineering \\
							Soongsil University, Seoul 06978, Korea \\
							\email{osshin@ssu.ac.kr}
}


\maketitle

\begin{abstract}
We consider the uplink of massive multiple-input multiple-output systems in a multicell environment. Since the base station (BS) estimates the channel state information (CSI) using the pilot signals transmitted from the users, each BS will have imperfect CSI in practice. Assuming zero-forcing method to eliminate the multi-user interference, we derive the exact analytical expressions for the probability density function (PDF) of the signal-to-interference-plus-noise ratio (SINR), the corresponding achievable rate, the outage probability, and the symbol error rate (SER) when the BS has imperfect CSI. An upper bound of the SER is also derived for an arbitrary number of antennas at the BS. Moreover, we derive the upper bound of the achievable rate for the case where the number of antennas at the BS goes to infinity, and the analysis is verified by presenting numerical results.
\keywords{Massive multiple-input multiple-output (MIMO) \and channel estimation \and channel state information (CSI) \and zero-forcing (ZF)}
\end{abstract}

\section{Introduction}
\label{section: Introduction}

Over the past decade, multiple-input multiple-output (MIMO) techniques have been investigated, starting from single-user scenarios \cite{GJJV:03:JSAC}. In point-to-point (single user) MIMO with $n_T$ transmit and $n_R$ receive antennas, the channel capacity is known to grow linearly as $\min(n_T,n_R)$ increases through the combined use of space-time coding and spatial multiplexing. A multi-user MIMO (MU-MIMO) scenario requires spatial sharing of the channel among the users connected to the network \cite{GKHCS:07:MAG,AWH:07:WCOM}. MU-MIMO has been discussed extensively for use in the 3GPP LTE-Advanced standardization. Relative to single-user MIMO (SU-MIMO), MU-MIMO has several key advantages in that it allows for direct gain in a multiple access capacity and spatial multiplexing gain at the base station (BS) without the need for multiple antenna terminals.

Recently, industry and academia have expressed significant interest in implementing massive MIMO in both single cell and multicell environments \cite{ HLM:13:COM, YM:13:JSAC, HBD:13:JSAC, MARZ:10:WCOM,Hanif:14:TCOM }. The use of additional antennas at the BS has been shown to improve power efficiency both for the uplink \cite{HLM:13:COM} and for the downlink \cite{YM:13:JSAC,Hanif:14:TCOM}. Massive MIMO is a system where a BS equipped with a hundred or more antennas simultaneously serves several users in the
same frequency band by exploiting the degrees-of-freedom (DoF) in the spatial domain \cite{Hanif:14:TCOM,HCPR:12:WCOM, OLT:13:JSAC,PML:12:COML,JAT:11:WCOM,MARZ:10:WCOM,YM:13:JSAC,HBD:13:JSAC,HLM:13:COM}. In such a system, many techniques were proposed to cancel and/or mitigate interference in MIMO systems, such as maximum likelihood multiuser detection in \cite{Dai} and a combination of a zero-forcing (ZF) and dirty paper coding (DPC) in \cite{Tran}. However, the complexity issue in practical systems for these techniques is a growing concern  proportional to the number of antennas at the BS \cite{Tran}. Fortunately, when the number of antennas at the BS is large enough, from the law of large numbers, the random and mutually independent channel vectors between the BS and the users become pairwise orthogonal \cite{HC:70:Book}. Indeed, the BS can purposefully avoid transmitting toward certain directions, and simple matched filer processing can be used to completely eliminate independent background noise and intracell interference \cite{MARZ:10:WCOM,HLM:13:COM,NMDL:13:VT,YGFL:13:JSAC}. Spatial-division multiplexing for massive MIMO can enhance the reliability and data rate of the system because more distinct paths are established between the BS and the users \cite{YGFL:13:JSAC,NMDL:13:VT}. Notably, the additional DoF provided by a massive number of antennas at the BS can reduce the transmit power for the users on the uplink. This is very efficient when multimedia services are increasing and the design of a battery with long time use is a major challenge for manufacturers \cite{Fehske}. Of course, the electrical power supply to the BS will be higher which is consumed by rectifier, baseband digital signal processing circuit, power amplifier, feeder, and cooling system on the downlink. Hence, solutions to reduce the emission of RF power would help in cutting down the power consumption of the BS \cite{Gonzalez}. Therefore, massive MIMO is a promising technology that can be integrated into next generation wireless systems \cite{RPLLM:13:MAG, LTEM:14:MAG}.

It is known that a multi-user detection technique called \, successive \, interference \, cancellation (SIC) can \,                              achieve maximum rate in the uplink \cite{Tse}. However,     the SIC is difficult to implement in practice due to its high computational complexity. So other detection methods that are based on linear detectors, including maximal ratio combining (MRC), ZF, and minimum mean square error (MMSE), have been developed \cite{YM:13:JSAC,NMDL:13:VT,NGO:13:COM,MARZ:10:WCOM,HBD:13:JSAC}. Among them, \cite{MARZ:10:WCOM} derived the asymptotic analysis for the signal-to-interference-plus-noise ratio (SINR) for the uplink by using MRC and the SINR for the downlink by using maximum-ratio transmission. An exact performance analysis for the uplink was provided in \cite{NMDL:13:VT} with arbitrary antennas at the BS. All of these results have shown that a linear receiver can exploit the advantages of massive antenna arrays at the BS with low implementation complexity. A ZF receiver can cancel intracell interference, and therefore it generally outperforms an MRC receiver. This implies that a ZF receiver can reduce the number of BS antennas necessary, relative to the number needed for MRC, while obtaining the same system performance. In general, the performance of the ZF receiver is worse than that of the MMSE receiver. However, if the SINR is high enough, the performance of the ZF receiver and that of the MMSE receiver are then equivalent \cite{HLM:13:COM, Tse}. Furthermore, an MMSE requires additional knowledge on the SINR and higher complexity than the ZF receiver. In addition, exact performance analysis is not tractable even in the case of perfect CSI \cite{HBD:13:JSAC}. It was shown in \cite{Hanif:14:TCOM} that the ZF beamforming can provide a good tradeoff between complexity and system performance, especially when the number of BS antennas is very large. Therefore, we focus on the ZF receiver in this paper.

In \cite{NMDL:13:VT}, the authors considered a similar MIMO configuration with a ZF receiver, where the CSI is assumed to be perfectly known to both the transmitter and the receiver. Under such assumptions for the CSI, the expression of the exact performance of the system might be tractable. In practice, however, CSI is not perfect at the transmitter and the receiver. In order for the BS to acquire the CSI, a simple scheme can be employed where users send pilot signals to the BS, so that the BS can estimate the channel by analyzing the received pilot signals in an uplink training phase \cite{HBD:13:JSAC,NGO:13:COM, JAT:11:WCOM,YGFL:13:JSAC,GJ:11:SPAWC}. The least-squares (LS) method is a conventional method that is generally used to estimate the channel state. Unfortunately, this method causes significant degradation in the system performance due strong intercell interference. In contrast, the MMSE estimation method can result in ​​more accurate channel estimation \cite{HBD:13:JSAC}. In the uplink transmission phase, the signals transmitted from the users to the BS can be detected by using a linear detector using the estimated CSI.

In this paper, we derive new results for the achievable rate of the uplink in a MIMO system with a ZF receiver. Our analysis considers a scenario that is more realistic than that presented in \cite{NMDL:13:VT}, since we consider that the channel state information at the BS is not perfect. Our main results are listed as follows.

\begin{itemize} 
\item In Section~\ref{SINR}, the probability density function (PDF) and cumulative distribution function (CDF) of the SINR are derived. These results are then used to analyze the bit error rate (BER) and the outage probability. We demonstrate how the performance depends only on the gain of the interfering links and how it saturates as the number of antennas at the BS becomes very large. An asymptotic analysis is also presented.  
\item In Section~\ref{Performance},  we derive the exact closed-form expressions of the achievable uplink rate, of the outage probability, and of the symbol error rate (SER) for an arbitrary number of antennas at the BS. In addition, an upper bound is also derived for the SER to evaluate the performance of the system in the absence of high precision requirements.
\end{itemize}

The rest of this paper is organized as follows. Section~\ref{System Model} describes the system model for massive MIMO in a multicell environment, and it explains how to estimate the channel state by using uplink pilot signals and how to use the estimated channel in a ZF receiver. In Section~\ref{SINR}, the statistics for the SINR are derived, and an asymptotic analysis is presented for the case where the number of antennas at the BS grows without limit. Section~\ref{Performance} analyzes the achievable uplink rate, the outage probability, and the SER for an arbitrary number of antennas at the BS. Section~\ref{Simulation} presents the numerical results that are used to evaluate the performance of the system. Finally, the conclusion is presented in Section~\ref{Conclusion}.

\textsl{Notation}: The superscript $\dagger$, *, $T$, and $\tr(\cdot)$ stand for the complex conjugate-transpose, the conjugate, the transpose, and the trace, respectively. $[\mathbf{G}]_{ij}$ denotes the entry of matrix $\mathbf{G}$ on the $i$-th row and $j$-th column, and $\mathbf{I}_{n}$ is the $n\times n$ identity matrix. We use $a \sim b$ to indicate that the random variables, $a$ and $b$, follow the same distribution. $\mathbb{E}\left\{{\cdot}\right\}$ denotes the expectation operation, and  $\mathbf{x} \sim \mathcal{CN} \left(\mathbf{0},\mathbf{\Sigma}\right)$ indicates that $\mathbf{x}$ is a symmetric complex Gaussian random vector with zero-mean and covariance matrix $\mathbf{\Sigma}$.

\section{System Model}
\label{System Model}

\subsection{Multicell Multi-user MIMO System}

In this paper, we consider the uplink transmission of a multicell multi-user MIMO system with $L > 1$ cells. Each cell is comprised of one BS and $K$ users. The BS's are assumed to be equipped with $M$ antennas and to use the same time-frequency resources to serve their own $K$ users. We assume $M \geq K$. Each user is assumed to be equipped with a single antenna, and the channel matrices between the BS and the users have to be estimated at the BS by using the uplink pilot signals. The $M \times 1$ received signal vector at the $\ell$-th BS is given as
\begin{equation}
\mathbf{y}_{\ell} = \sqrt{P_{\mathrm{u}}}\sum_{i=1}^{L}\mathbf{G}_{i\ell}\mathbf{x}_{i} + \mathbf{n}_{\ell},
\label{1}
\end{equation}
where $\mathbf{G}_{i\ell}$ represents the $M \times K$ channel matrix between the $\ell$-th BS and the $K$ users in the $i$-th cell,  $P_{\mathrm{u}}$ is the average transmit power of each user, and $\mathbf{n}_{\ell}\sim\mathcal{CN}\left(\mathbf{0}, \mathbf{I}_M\right)$ is a vector of the additive white Gaussian noise (AWGN). $\mathbf{x}_{i}$ is a $K\times1$ vector of message-bearing quadrature amplitude modulation (QAM) symbols transmitted from the $K$ users to the BS in the $i$-th cell, and the average power of each symbol is normalized to that of unity:
\begin{equation}
\mathbb{E}\left\{\mathbf{x}_{i}\mathbf{x}_{i}^{\dagger}\right\} = \mathbf{I}_{K}.
\label{xi}
\end{equation}

The channel matrix $\mathbf{G}_{i\ell}$ can be decomposed as
\begin{equation}
\mathbf{G}_{i\ell} = \mathbf{H}_{i\ell}\mathbf{D}_{i\ell}^{1/2},
\label{2}
\end{equation}
where $\mathbf{H}_{i\ell}$ is an $M \times K$ matrix of which the $(m,k)$ element $\left[\mathbf{H}_{i\ell}\right]_{mk}$ is the fast-fading coefficient from the $k$-th user in the $i$-th cell to the $m$-th antenna at the BS in the $\ell$-th cell. We assume that the entries for $\mathbf{H}_{i\ell}$ are independent and identically distributed (i.i.d.) complex Gaussian random variables of zero mean and unit variance. $\mathbf{D}_{i\ell}$ is a $K \times K$ diagonal matrix for which $\beta_{i\ell,k} \triangleq \left[\mathbf{D}_{i\ell}\right]_{kk}$ represents the large-scale fading incorporating the path loss and the shadowing. Accordingly, the $(m,k)$ element of $\mathbf{G}_{i\ell}$ can be expressed as
\begin{equation}
\begin{array}{lcl}
\left[\mathbf{G}_{i\ell}\right]_{mk} &=& \sqrt{\beta_{i\ell,k}} \left[\mathbf{H}_{i\ell}\right]_{mk}, 
\\ && m = 1, 2, \cdots, M, \quad k = 1, 2, \cdots, K.
\end{array}
\label{3}
\end{equation}

\subsection{Uplink Channel Estimation}

In order to detect the data transmitted by $K$ users in the cell, the BS needs to estimate the CSI for the users. Let $\tau_{\mathrm{u}}$ denote the number of symbols per coherence interval used for the uplink training, and let $T$ be the length of the coherence interval. All users in every cell simultaneously transmit $\tau_{\mathrm{u}}$ orthogonal uplink pilot symbols. In our work, we assume $T>\tau_{\mathrm{u}}\geq K$. Let $\sqrt{\tau_{\mathrm{u}}P_{\mathrm{u}}}\mathbf{\Psi}_{\ell}$ be the pilot sequences transmitted by $K$ users in the $\ell$-th cell. Note that $\mathbf{\Psi}_{\ell}$ is a $\tau_\mathrm{u} \times K$ matrix, and it satisfies $\mathbf{\Psi}_{\ell}^{\dagger}\mathbf{\Psi}_{\ell} = \mathbf{I}_{K}$, so that the intracell interference between the pilot signals can be ignored. We also assume that all the users in the $L$ cells use the same set of pilot sequences.  Similarly to \eqref{1}, the received pilot matrix at the $\ell$-th BS is given as

\begin{equation}
\mathbf{Y}_{p,\ell} = \sqrt{{\tau_{\mathrm{u}}P_{\mathrm{u}}}}\sum_{i=1}^{L}\mathbf{G}_{i\ell}\mathbf{\Psi}_{i}^{T} + \textbf{N}_{\ell},
\label{4}
\end{equation}
where $\mathbf{N}_{\ell}$ is an $M \times \tau_{\mathrm{u}}$ AWGN matrix. Note that $\mathbf{G}_{\ell\ell}$ is the desired channel for the own cell, while $\mathbf{G}_{i\ell}$'s, $i \neq \ell$, are the interference channels from the other cells. The received signal $\mathbf{Y}_{p,\ell}$ projected onto $\mathbf{\Psi}_{\ell}^{*}$ is given as
\begin{equation}
\tilde{\mathbf{Y}}_{p,\ell} = \mathbf{Y}_{p,\ell}\mathbf{\Psi}_{\ell}^{*} = \sqrt{{\tau_{\mathrm{u}}P_{\mathrm{u}}}} \sum_{i=1}^{L}\mathbf{G}_{i\ell} + \mathbf{W}_{\ell},
\label{5}
\end{equation}
where $\mathbf{W}_{\ell} \triangleq \mathbf{N}_{\ell}\mathbf{\Psi}_{\ell}^{*}$. Let $\tilde{\mathbf{y}}_{p,\ell k}$ denote the $k$-th column of $\tilde{\mathbf{Y}}_{p,\ell}$ as
\begin{equation}
\tilde{\mathbf{y}}_{p,\ell k} =  \sqrt{{\tau_{\mathrm{u}}P_{\mathrm{u}}}} \sum_{i=1}^{L}\mathbf{h}_{i\ell, k}\sqrt{\beta_{i\ell, k}} + \mathbf{w}_{\ell, k},
\label{eq:column of Y}
\end{equation}
where $\mathbf{h}_{i\ell, k}$ and $\mathbf{w}_{\ell, k}$ are the $k$-th columns of the $\mathbf{H}_{i\ell}$ and $\mathbf{W}_{\ell}$, respectively.
The MMSE, or equivalently Bayesian estimate, for the channel matrix $\mathbf{h}_{\ell\ell, k}$ is given as~\cite[Eq.~(12.7)]{KAY:93:Book}
\begin{equation} 
\hat{\mathbf{h}}_{\ell\ell, k} = \sqrt{{\tau_{\mathrm{u}}P_{\mathrm{u}}}\beta_{\ell\ell,k}}\Bigl(\tau_{\mathrm{u}}P_{\mathrm{u}}\sum_{i=1}^{L}\beta_{i\ell,k}+1\Bigr)^{-1}\tilde{\mathbf{y}}_{p,\ell k}.
\label{eq:Estimation}
\end{equation}
Thus, the estimate $\hat{\mathbf{H}}_{\ell\ell}$ of $\mathbf{H}_{\ell\ell}$ is given as
\begin{equation} 
\hat{\mathbf{H}}_{\ell\ell} = \frac{1}{\sqrt{\tau_{\mathrm{u}}P_{\mathrm{u}}}}\tilde{\mathbf{Y}}_{p,\ell}\mathbf{D}_{\ell}^{-1}\mathbf{D}_{\ell\ell}^{1/2},
\label{eq:Estimation_H}
\end{equation}
where $\mathbf{D}_{\ell}\triangleq\sum_{i=1}^{L}\mathbf{D}_{i\ell}+\frac{1}{\tau_\mathrm{u}P_{\mathrm{u}}}\mathbf{I}_K$. By multiplying both sides of \eqref{eq:Estimation_H} with $\mathbf{D}_{\ell\ell}^{1/2}$, we have the estimate of the channels between the BS and the users in the $\ell$-th cell as
\begin{equation} \hat{\mathbf{G}}_{\ell\ell} = \Bigl(\sum_{i=1}^{L}\mathbf{G}_{i\ell} + \frac{1}{\sqrt{\tau_{\mathrm{u}}P_\mathrm{u}}}\mathbf{W}_{\ell}\Bigr)\hat{\mathbf{D}}_{\ell\ell},
\label{6}
\end{equation}
or equivalently
\begin{equation} \hat{\mathbf{G}}_{\ell\ell} = \mathbf{G}_{\ell\ell}\hat{\mathbf{D}}_{\ell\ell}+\sum_{i=1,i\neq\ell}^{L}\mathbf{G}_{i\ell}\hat{\mathbf{D}}_{\ell\ell} + \frac{1}{\sqrt{\tau_{\mathrm{u}}P_\mathrm{u}}}\mathbf{W}_{\ell}\hat{\mathbf{D}}_{\ell\ell},
\label{6a}
\end{equation}
where $\hat{\mathbf{D}}_{\ell\ell}\triangleq\mathbf{D}_{\ell}^{-1}\mathbf{D}_{\ell\ell}$ is a diagonal matrix with the $k$-th diagonal element $\bigl[\hat{\mathbf{D}}_{\ell\ell}\bigr]_{kk} = \beta_{\ell\ell,k}\left(\sum_{i=1}^{L}\beta_{i\ell,k}+\frac{1}{\tau_{\mathrm{u}}P_\mathrm{u}}\right)^{-1}$. The second term in \eqref{6a} is the channel estimation error caused by the users from the other cells, which is called pilot contamination \cite{NGO:13:COM,MARZ:10:WCOM}. However, the BSs cannot distinguish pilot signals of intra-cell users from those of inter-cell users, since each BS is assigned the same pilot sequences. This causes the achievable performance to be saturated.

In a similar way to \eqref{6}, we can estimate $\mathbf{G}_{ii}$ for $i\neq\ell$. Let $\xi_{ii}, i = 1,2,\cdots,L$, denote the channel estimation error for $\mathbf{G}_{ii}$. Then, $\mathbf{G}_{ii}$ can be expressed as
\begin{equation}
 \mathbf{G}_{ii}=\hat{\mathbf{G}}_{ii}+\mathbf{\xi}_{ii}.
\label{7}
\end{equation}
We assume that the error $\mathbf{\xi}_{ii}$ is independent of $\hat{\mathbf{G}}_{ii}$ \cite{ KAY:93:Book}.

\subsection{Linear ZF Receiver}

After the BS estimates the CSI using the pilot signals, the data transmission of the $K$ users in the cell begins. The BS uses the estimated CSI to detect the independent transmit data streams from the users. In this subsection, we describe how a ZF receiver works to detect the data streams. Let $\mathbf{A}_{\ell}$ be the ZF receiver matrix, i.e.,
\begin{equation}
\mathbf{A}_{\ell} = \hat{\mathbf{G}}_{\ell\ell}\left(\hat{\mathbf{G}}_{\ell\ell}^{\dagger}\hat{\mathbf{G}}_{\ell\ell}\right)^{-1}.
\label{8}
\end{equation}
From \eqref{1} and \eqref{8}, the received uplink signal at the $\ell$-th BS is separated into $K$ data streams by multiplying $\mathbf{y}_{\ell}$ with $\mathbf{A}_{\ell}^{\dagger}$ as
\begin{equation}
\begin{array}{lcl}
	  \mathbf{r}_{\ell} &=& \mathbf{A}_{\ell}^{\dagger}\mathbf{y}_{\ell}= \sqrt{P_{\mathrm{u}}} \mathbf{x}_{\ell} + \sqrt{P_{\mathrm{u}}}\mathbf{A}_{\ell}^{\dagger} \mathbf{\xi}_{\ell\ell} \mathbf{x}_{\ell} \\
	   &+& \sqrt{P_{\mathrm{u}}}\sum_{i=1,i\neq \ell}^{L} \mathbf{A}_{\ell}^{\dagger} \mathbf{G}_{i\ell} \mathbf{x}_{i} + \mathbf{A}_{\ell}^{\dagger} \mathbf{n}_{\ell}.
\end{array}
\label{9}
\end{equation} 
Correspondingly, the signal received for the $k$-th user at the $\ell$-th BS is derived as
\begin{equation}
\begin{array}{lcl}
	  \mathbf{r}_{\ell,k} &=& \sqrt{P_{\mathrm{u}}} \mathbf{x}_{\ell,k} +  \sqrt{P_{\mathrm{u}}}\mathbf{a}_{\ell,k}^{\dagger} \mathbf{\xi}_{\ell\ell} \mathbf{x}_{\ell} \\
	   &+& \sqrt{P_{\mathrm{u}}}\sum_{i=1,i\neq \ell}^{L} \mathbf{a}_{\ell,k}^{\dagger} \mathbf{G}_{i\ell} \mathbf{x}_{i} 
	   + \mathbf{a}_{\ell,k}^{\dagger} \mathbf{n}_{\ell},
\end{array}
\label{10}
\end{equation}
where $\mathbf{a}_{\ell,k}$ denotes the $k$-th column of the matrix $\mathbf{A}_{\ell}$, and $\mathbf{x}_{\ell,k}$ denotes the $k$-th element of the vector $\mathbf{x}_{\ell}$. The SINR $\mathbf{\gamma}_{\ell,k}$ of the $k$-th user in the $\ell$-th cell after the ZF processing can be computed as
\begin{equation}
\begin{aligned}
\mathbf{\gamma}_{\ell,k} =&\Bigl(\mathbb{E}\Bigl\{\left|\mathbf{a}_{\ell,k}^{\dagger}\mathbf{\xi}_{\ell\ell}\mathbf{x}_{\ell}\right|^{2}\Bigr\}+\mathbb{E}\Bigl\{\sum_{i=1,i\neq \ell}^{L}\left|\mathbf{a}_{\ell,k}^{\dagger}\mathbf{G}_{i\ell}\mathbf{x}_{i}\right|^2 \Bigr\}\\
& +\mathbb{E}\Bigl\{\frac{\left|\mathbf{a}_{\ell,k}^{\dagger}\mathbf{n}_{\ell}\right|^{2}}{P_{\mathrm{u}}}\Bigr\}\Bigr)^{-1}.
\end{aligned}
\label{12}
\end{equation}

\section{PDF of Uplink SINR and Asymptotic Analysis}
\label{SINR}

\subsection{PDF of Uplink SINR}

In this subsection, we derive the PDF of the uplink SINR for the ZF detection. From \eqref{12}, the SINR of the $k$-th user in $\ell$-th cell can be expressed as
\begin{equation}
\begin{aligned}
&\mathbf{\gamma}_{\ell,k} = \\
&\frac{1}{\left(\sum_{i=1}^{L}\alpha_{i\ell}+\frac{1}{P_{\mathrm{u}}}\right)\Bigl[\left(\hat{\mathbf{G}}_{\ell\ell}^{\dagger}\hat{\mathbf{G}}_{\ell\ell}\right)^{-1}\Bigr]_{kk}+\frac{\sum_{i=1,i\neq \ell}^{L}\beta_{i\ell,k}^{2}}{\beta_{\ell\ell,k}^{2}}},
\end{aligned}
\label{15}
\end{equation}
where $\alpha_{i\ell} \triangleq \mathrm{tr}\left(\mathbb{E}\left[\xi_{i\ell}\xi_{i\ell}^{\dagger}\right]\right) = \sum_{k=1}^{K}\frac{\tau_{\mathrm{u}}P_{\mathrm{u}}\beta_{i\ell,k}\sum_{j=1,j\neq i}^{L}\beta_{j\ell,k}}{\tau_{\mathrm{u}}P_{\mathrm{u}}\sum_{j=1}^{L}\beta_{j\ell,k}+1}$ can be obtained by using~\cite[Eq.~(12.8)]{KAY:93:Book}. The derivation of \eqref{15} can be proved from
\begin{equation}
\begin{aligned}
& \mathbb{E}\Bigl\{\left|\mathbf{a}_{\ell,k}^{\dagger}\mathbf{\xi}_{\ell\ell}\mathbf{x}_{\ell}\right|^{2}\Bigr\}=\mathbb{E}\left\{\mathbf{a}_{\ell,k}^{\dagger}\mathbf{\xi}_{\ell\ell}\mathbf{x}_{\ell}\mathbf{x}_{\ell}^{\dagger}\mathbf{\xi}_{\ell\ell}^{\dagger}\mathbf{a}_{\ell,k}\right\}\\
&=\tr\left(\mathbb{E}\left[\mathbf{\xi}_{\ell\ell}\mathbf{\xi}_{\ell\ell}^{\dagger}\right]\right)\left|\mathbf{a}_{\ell,k}^{\dagger}\right|^{2}=\alpha_{\ell\ell}\Bigl[\left(\hat{\mathbf{G}}_{\ell\ell}^{\dagger}\hat{\mathbf{G}}_{\ell\ell}\right)^{-1}\Bigr]_{kk},
\end{aligned}
\label{16}
\end{equation}

\begin{equation}
\begin{aligned}
 &\mathbb{E}\Bigl\{\sum_{i=1,i\neq \ell}^{L}\left|\mathbf{a}_{\ell,k}^{\dagger}\mathbf{G}_{i\ell}\mathbf{x}_{i}\right|^2 \Bigr\}\\
&=\mathbb{E}\Bigl\{\sum_{i=1,i\neq \ell}^{L}\left|\mathbf{a}_{\ell,k}^{\dagger}\left\{\hat{\mathbf{G}}_{i\ell}+\mathbf{\xi}_{i\ell}\right\}\mathbf{x}_{i}\right|^2 \Bigr\}\\
&=\mathbb{E}\Bigl\{\sum_{i=1,i\neq \ell}^{L}\mathbf{a}_{\ell,k}^{\dagger}\hat{\mathbf{G}}_{i\ell}\mathbf{x}_{i}\mathbf{x}_{i}^{\dagger}\hat{\mathbf{G}}_{i\ell}^{\dagger}\mathbf{a}_{\ell,k}\Bigr\}\\
&\qquad+ \mathbb{E}\Bigl\{\sum_{i=1,i\neq \ell}^{L}\mathbf{a}_{\ell,k}^{\dagger}\mathbf{\xi}_{i\ell}\mathbf{x}_{i}\mathbf{x}_{i}^{\dagger}\mathbf{\xi}_{i\ell}^{\dagger}\mathbf{a}_{\ell,k}\Bigr\}\\
&\stackrel{(a)}{=}\frac{\sum_{i=1,i\neq \ell}^{L}\beta_{i\ell,k}^2}{\beta_{\ell\ell,k}^2}\mathbb{E}\left\{\mathbf{a}_{\ell,k}^{\dagger}\hat{\mathbf{G}}_{\ell\ell}\hat{\mathbf{G}}_{\ell\ell}^{\dagger}\mathbf{a}_{\ell,k}\right\}\\ 
&\qquad+\sum_{i=1,i\neq \ell}^{L}\tr\left(\mathbb{E}\left[\mathbf{\xi}_{i\ell}\mathbf{\xi}_{i\ell}^{\dagger}\right]\right)\left|\mathbf{a}_{\ell,k}^{\dagger}\right|^{2} \\
&=\frac{\sum_{i=1,i\neq \ell}^{L}\beta_{i\ell,k}^2}{\beta_{\ell\ell,k}^2}\mathbb{E}\left\{\left[\mathbf{A}_{\ell}^{\dagger}\hat{\mathbf{G}}_{\ell\ell}\hat{\mathbf{G}}_{\ell\ell}^{\dagger}\mathbf{A}_{\ell}\right]_{kk}\right\}\\ 
&\qquad+\sum_{i=1,i\neq \ell}^{L}\alpha_{i\ell}\left[\left(\hat{\mathbf{G}}_{\ell\ell}^{\dagger}\hat{\mathbf{G}}_{\ell\ell}\right)^{-1}\right]_{kk}\\ 
&\stackrel{(b)}{=}\frac{\sum_{i=1,i\neq \ell}^{L}\beta_{i\ell,k}^{2}}{\beta_{\ell\ell,k}^{2}} + \sum_{i=1,i\neq \ell}^{L}\alpha_{i\ell}\left[\left(\hat{\mathbf{G}}_{\ell\ell}^{\dagger}\hat{\mathbf{G}}_{\ell\ell}\right)^{-1}\right]_{kk},
\end{aligned}
\label{17}
\end{equation}

\noindent and
\begin{equation}
\mathbb{E}\Bigl\{\frac{1}{P_{\mathrm{u}}}\left|\mathbf{a}_{\ell,k}^{\dagger}\mathbf{n}_{\ell}\right|^{2}\Bigr\}=\frac{1}{P_{\mathrm{u}}}\left|\mathbf{a}_{\ell,k}^{\dagger}\right|^{2}=\frac{1}{P_{\mathrm{u}}}\left[\left(\hat{\mathbf{G}}_{\ell\ell}^{\dagger}\hat{\mathbf{G}}_{\ell\ell}\right)^{-1}\right]_{kk}.
\label{18}
\end{equation}
Based on the results in \eqref{6} and \eqref{7}, let $\hat{\mathbf{g}}_{i\ell,k}$ and $\mathbf{\xi}_{i\ell,k}$ be the columns of $\hat{\mathbf{G}}_{i\ell}$ and $\mathbf{\xi}_{i\ell}$, respectively, where $\hat{\mathbf{g}}_{i\ell,k}\sim\mathcal{CN}\left(\mathbf{0},\frac{\beta_{i\ell,k}^{2}}{\hat{\beta}_{i k}}\right)$, with $\hat{\beta}_{i k} \triangleq \sum_{j=1}^{L}\beta_{j\ell,k}+\frac{1}{\tau_{\mathrm{u}}P_{\mathrm{u}}}$, and $\mathbf{\xi}_{i\ell,k}\sim\mathcal{CN}\left(\mathbf{0},\left(\beta_{i\ell,k}-\hat{\beta}_{i k}\right)\mathbf{I}_M\right)$. Therefore, $\alpha_{i\ell}, i=1,2,\cdots, L$ in \eqref{16} and \eqref{17} can easily be derived. From the fact that $\hat{\mathbf{g}}_{i\ell,k}=\sqrt{\beta_{i\ell,k}}\hat{\mathbf{h}}_{i\ell,k}$ and \eqref{eq:Estimation}, we have $\hat{\mathbf{g}}_{i\ell,k}=\frac{\beta_{i\ell,k}}{\beta_{\ell\ell,k}}\hat{\mathbf{g}}_{\ell\ell,k}$ \cite[eq. (51)]{NGO:13:COM},  which leads to the equality $(a)$ in \eqref{17}. Furthermore, we have
\begin{equation}
\begin{aligned}
\mathbf{A}_{\ell}^{\dagger}\hat{\mathbf{G}}_{\ell\ell}\hat{\mathbf{G}}_{\ell\ell}^{\dagger}\mathbf{A}_{\ell}&=\bigl(\hat{\mathbf{G}}_{\ell\ell}^{\dagger}\hat{\mathbf{G}}_{\ell\ell}\bigr)^{-1}\hat{\mathbf{G}}_{\ell\ell}^{\dag}\hat{\mathbf{G}}_{\ell\ell}\hat{\mathbf{G}}_{\ell\ell}^{\dagger} \hat{\mathbf{G}}_{\ell\ell}\bigl(\hat{\mathbf{G}}_{\ell\ell}^{\dagger}\hat{\mathbf{G}}_{\ell\ell}\bigr)^{-1}\\
                  &=\mathbf{I}_K.
\end{aligned}
\notag
\end{equation}
Then, the equality $(b)$ in \eqref{17} can be derived.
Substituting \eqref{16}, \eqref{17} and \eqref{18} into \eqref{12}, we obtain \eqref{15}.

Let $X \triangleq \frac{1}{\left[\left(\hat{\mathbf{G}}_{\ell\ell}^{\dagger}\hat{\mathbf{G}}_{\ell\ell}\right)^{-1}\right]_{kk}}$, then \eqref{15} can be rewritten as

\begin{equation}
\mathbf{\gamma}_{\ell,k} =\frac{X}{\left(\sum_{i=1}^{L}\alpha_{i\ell}+\frac{1}{P_{\mathrm{u}}}\right)+\frac{\sum_{i=1,i\neq \ell}^{L}\beta_{i\ell,k}^{2}}{\beta_{\ell\ell,k}^{2}}X}.
\label{21}
\end{equation}
Note that from \eqref{6},  $\hat{\mathbf{G}}_{\ell\ell}^{\dagger}\hat{\mathbf{G}}_{\ell\ell}$ is a central complex Wishart matrix with $M$ degrees of freedom and covariance matrix $\mathbf{\Sigma}_{\ell\ell}=\mathrm{diag}\left\{\frac{\beta_{\ell\ell,1}^2}{\hat{\beta}_{\ell 1}},\cdots,\frac{\beta_{\ell\ell,K}^2}{\hat{\beta}_{\ell K}}\right\}$ where $\hat{\beta}_{\ell k} \triangleq \sum_{j=1}^{L}\beta_{j\ell,k}$  $+\frac{1}{\tau_{\mathrm{u}}P_{\mathrm{u}}}$ as in \cite{Tulino}. Thus, $X$ follows a complex central Wishart distribution with $M-K+1$ degrees of freedom and scale parameter $\frac{\beta_{\ell\ell,k}^{2}}{\hat{\beta}_{\ell k}}$. Therefore, the PDF of $X$ is given as \cite{ GHP:02:COML}
\begin{equation}
p_{X}(x) = \frac{e^{-x\hat{\beta}_{\ell k}/\beta_{\ell\ell,k}^{2}}}{(M-K)!\beta_{\ell\ell,k}^{2}/\hat{\beta}_{\ell k}}\left(\frac{x\hat{\beta}_{\ell k}}{\beta_{\ell\ell,k}^{2}}\right)^{M-K}, \quad x>0.
\label{20}
\end{equation}

\lbn
\begin{proposition}
\label{Proposition1}
The PDF of the uplink SINR for the $k$-th user in the $\ell$-th cell with the ZF receiver is given as
\begin{equation}
\begin{aligned}
&p_{\mathbf{\gamma}_{\ell,k}}(s)\\
&= \theta^{M-K+1}\frac{ s^{M-K}}{\left(1-\eta s\right)^{M-K+2}}\exp\left(-\frac{\theta s}{1-\eta s}\right),\quad s<\frac{1}{\eta},
\end{aligned}
\label{22}
\end{equation}
\noindent where $\theta \triangleq \sum_{i=1}^{L}\alpha_{i,\ell} + \frac{1}{P_{\mathrm{u}}}$ and $\eta \triangleq \frac{\sum_{i=1,i\neq \ell}^{L}\beta_{i\ell,k}^{2}}{\beta_{\ell\ell,k}^{2}}$.
\end{proposition}

\lbn
\begin{proof}
From \eqref{21}, we can compute the cumulative distribution function (CDF) of $\mathbf{\gamma}_{\ell,k}$ as
\begin{equation}
\begin{aligned}
F_{\mathbf{\gamma}_{\ell,k}}(s) = \Pr\left(\frac{X}{\theta+\eta X}<s\right) 
= \Pr\left(X<\frac{\theta s}{1-\eta s}\right).
\end{aligned}
\label{24}
\end{equation}

\begin{itemize}
  \item If $s < 1/\eta$, by using ~\cite[Eq.~(3.381.1)]{GR:07:Book}, \eqref{24} can be written as \\
	\begin{equation}
	\begin{aligned}
	F_{\mathbf{\gamma}_{\ell, k}}(s)&=\int_{0}^{\frac{\theta s}{1-\eta s}}\frac{e^{-x\hat{\beta}_{\ell k}/\beta_{\ell\ell,k}^{2}}}{(M-K)!\beta_{\ell\ell, k}^{2}/\hat{\beta}_{\ell k}}\left(\frac{x\hat{\beta}_{\ell k}}{\beta_{\ell k}^{2}}\right)^{M-K}dx\\
	&=\int_{0}^{\frac{\hat{\beta}_{\ell k}}{\beta_{\ell\ell,k}^{2}}\frac{\theta s}{1-\eta s}}\frac{e^{-\hat{x}}}{(M-K)!}\hat{x}^{M-K}d\hat{x}\\
	&=1-\exp\left({-\frac{\theta s}{1-\eta s}}\right)\sum_{i=0}^{M-K}\frac{1}{i!}\left(\frac{\theta s}{1-\eta s}\right)^{i},\\
   \end{aligned}
\label{25}
\end{equation}
\noindent where  $\hat{x}=\frac{\hat{\beta}_{\ell k}}{\beta_{\ell\ell,k}^{2}}x$.      
  \item If $s\geq 1/\eta$, \eqref{24} becomes\\
       \begin{equation}
			F_{\mathbf{\gamma}_{\ell,k}}(s)=1.
			\label{26}
      \end{equation} 
\end{itemize}
\noindent We can derive the PDF in \eqref{22} by applying the relationship $p_{\mathbf{\gamma}_{\ell,k}}(s) = \frac{dF_{\mathbf{\gamma}_{\ell,k}}(s)}{ds}$ between the CDF and the PDF to \eqref{25} and \eqref{26}, and by using ~\cite[Eq.~(8.356.4)]{GR:07:Book}. 
\end{proof}

\subsection{Asymptotic Analysis}

One of the key properties of the massive MIMO is that the number of BS antennas, $M$, can grow without limit, whereas the number of users, $K$, and the transmit power of each user are finite. The performance will improve as the number of antennas increases. However, the improvement will be limited by the intercell interference. In this subsection, we present asymptotic analysis of the SINR to get insight into the system behavior for the case of infinite number of antennas.

\lbn
\begin{proposition}
\label{Proposition2}
When $P_{\mathrm{u}}$ and $K$ are fixed, and the number of antennas at the BS, $M$, grows without bound, the effective uplink SINR of the $k$-th user in the $\ell$-th cell for the ZF receiver approaches the same value as that of the MRC receiver in \cite{ MARZ:10:WCOM}, which is given by

\begin{equation}
\mathrm{SINR}_{\ell,k}^{\infty}=\frac{\beta_{\ell\ell,k}^{2}}{\sum_{j=1,j\neq \ell}^{L}\beta_{j\ell,k}^{2}}.
\label{13}
\end{equation}
\end{proposition}

\lbn
\begin{proof}
From \eqref{1} and \eqref{8}, we have
\begin{small}
\begin{equation}
\begin{aligned}
&\mathbf{r}_{\ell} = \mathbf{A}_{\ell}^{\dagger}\mathbf{y}_{\ell} = 
\left( \hat{\mathbf{G}}_{\ell\ell}^{\dagger} \hat{\mathbf{G}}_{\ell\ell}\right)^{-1} 
 \hat{\mathbf{G}}_{\ell\ell}^{\dagger} \Bigl(\sqrt{P_{\mathrm{u}}}\sum_{j=1}^{L}\mathbf{G}_{j\ell}\mathbf{x}_{j} + \mathbf{n}_{\ell}\Bigr) \\
&=\left( \hat{\mathbf{G}}_{\ell\ell}^{\dagger} \hat{\mathbf{G}}_{\ell\ell}\right)^{-1} 
\hat{\mathbf{D}}_{\ell\ell} \Bigl(\sqrt{P_{\mathrm{u}}}\sum_{i=1}^{L}\sum_{j=1}^{L}\mathbf{G}_{i\ell}^{\dagger}\mathbf{G}_{j\ell} \mathbf{x}_{j} + \sum_{i=1}^{L}\mathbf{G}_{i\ell}^{\dagger}\mathbf{n}_{\ell}\\
&\quad\quad +\frac{1}{\sqrt{\tau_\mathrm{u}}}\sum_{j=1}^{L}\mathbf{W}_{\ell}^{\dagger}\mathbf{G}_{j\ell}\mathbf{x}_{j} +\frac{1}{\sqrt{\tau_{\mathrm{u}}P_{\mathrm{u}}}}\mathbf{W}_{\ell}^{\dagger} \mathbf{n}_{\ell}\Bigl).
\end{aligned}
\label{14}
\end{equation}
\end{small}
When $K$ is fixed and $M$ goes to infinity, we have \footnote{Assume that the $k$-th columns $\mathbf{h}_{i\ell,k}$ and $\mathbf{h}_{j\ell,k}$ of $\mathbf{H}_{i\ell}$ and $\mathbf{H}_{j\ell k}$, respectively, are mutually independent $M\times 1$ vectors, whose elements are i.i.d. random variables with zero mean and unit variance. From the law of large numbers, we have
$\frac{1}{M}\mathbf{h}_{i\ell,k}^{\dagger}\mathbf{h}_{i\ell,k}\stackrel{a.s}{\rightarrow}1$ and $\frac{1}{M}\mathbf{h}_{i\ell,k}^{\dagger}\mathbf{h}_{j\ell,k}\stackrel{a.s}{\rightarrow}0$ as $M\rightarrow\infty$, where $\stackrel{a.s}{\rightarrow}$ denotes the almost sure convergence.} \cite{ HLM:13:COM, HC:70:Book}
\begin{equation}
\begin{aligned}
&\frac{ \hat{\mathbf{G}}_{\ell\ell}^{\dagger} \hat{\mathbf{G}}_{\ell\ell}}{M} \stackrel{M\rightarrow\infty} {\rightarrow}\hat{\mathbf{D}}_{\ell\ell}\Bigl(\sum_{i=1}^{L}\mathbf{D}_{i\ell}+\frac{1}{\tau_{\mathrm{u}}P_{\mathrm{u}}}\mathbf{I}_{K}\Bigr)\hat{\mathbf{D}}_{\ell\ell},\\
&\sum_{i=1}^{L}\sum_{j=1}^{L}\frac{\mathbf{G}_{i\ell}^{\dagger}\mathbf{G}_{j\ell}}{M}\stackrel{M\rightarrow\infty}{\rightarrow}\mathbf{D}_{\ell\ell}+\sum_{j=1,j\neq \ell}^{L}\mathbf{D}_{j\ell},
\end{aligned}
\label{limit}
\end{equation}
and all the other terms in the last parenthesis of \eqref{14} converge to zero, where the notation $\stackrel{M\rightarrow\infty}{\rightarrow}$ is used to indicate that the number of antennas at the BS goes to infinity. 

\noindent Then, substituting \eqref{limit} into \eqref{14}, we obtain
\begin{equation}
\begin{aligned}
\mathbf{r}_{\ell}\stackrel{M\rightarrow\infty}{\rightarrow} &\left(\hat{\mathbf{D}}_{\ell\ell}\Bigl(\sum_{i=1}^{L}\mathbf{D}_{i\ell}+\frac{1}{\tau_{\mathrm{u}}P_{\mathrm{u}}}\mathbf{I}_{K}\Bigr)\hat{\mathbf{D}}_{\ell\ell}\right)^{-1}\hat{\mathbf{D}}_{\ell\ell}\\
&\times \Bigl(\sqrt{P_{\mathrm{u}}}\mathbf{D}_{\ell\ell}\mathbf{x}_{\ell}+\sqrt{P_{\mathrm{u}}}\sum_{j=1,j\neq \ell}^{L}\mathbf{D}_{j\ell}\mathbf{x}_{j}\Bigr)
\end{aligned}
\label{limit_sir}
\end{equation}
Therefore, the asymptotic uplink SINR for the $k$-th user in the $\ell$-th cell is obtained as \eqref{13}. 
\end{proof}
\lbn

We can see that the MRC receiver in \cite{MARZ:10:WCOM} and the ZF receiver in this paper exhibit the same performance when the BS is equipped with a very large number of antennas, i.e., as $M\rightarrow\infty$. \eqref{13} indicates that the effect of the interference does not vanish due to pilot contamination, even though the number of BS antennas goes to infinity. Furthermore, the SINR in \eqref{13} is independent of $P_{\mathrm{u}}$, which means that we can exploit a lower power regime while maintaining the limit SINR of \eqref{13} constant. However, in practice, the number of users and the transmit power of each user are not always fixed. For example, the following two scenarios in Corollary 1 and Corollary 2 should be considered.

\lbn
\begin{corollary}
When $P_{\mathrm{u}}$ is fixed, and $M$, $K$ grow without bound with $\mu\triangleq M/K >1$ remaining constant, the effective uplink SINR of the $k$-th user in the $\ell$-th cell in \eqref{15} can be deterministically approximated for the ZF receiver to
\begin{equation}
\begin{aligned}
&\overline{\mathrm{SINR}}_{\ell,k}\\
&=\frac{\beta_{\ell\ell,k}^2(\mu-1)}{\beta_{\ell\ell, k}^2\hat{\beta}_{\ell k}\sum_{i=1}^{L}\alpha_{i\ell}/K +\sum_{i=1,i\neq\ell}^{L}\beta_{i\ell,k}^2(\mu-1)}.
\label{eq:sinr1}
\end{aligned}
\end{equation}
\end{corollary}
\lbn

\begin{proof}
From  \eqref{21}, with   $\mathbb{E}\left\{X\right\}=\left(M-K+1\right)\frac{\beta_{\ell\ell,k}^{2}}{\hat{\beta}_{\ell k}}$, we have
\begin{equation}
\begin{aligned}
&\overline{\mathrm{SINR}}_{\ell,k}\geq \\
& \frac{\beta_{\ell\ell,k}^2(M-K+1)}{\beta_{\ell\ell, k}^2\hat{\beta}_{\ell k}\left(\sum_{i=1}^{L}\alpha_{i\ell}+\frac{1}{P_{\mathrm{u}}}\right) +\sum_{i=1,i\neq\ell}^{L}\beta_{i\ell,k}^2(M-K+1)}\\
&\stackrel{M,K\rightarrow\infty}{\rightarrow}\frac{\beta_{\ell\ell,k}^2(M/K-1)}{\beta_{\ell\ell, k}^2\hat{\beta}_{\ell k}\sum_{i=1}^{L}\frac{\alpha_{i\ell}}{K} +\sum_{,i=1,i\neq\ell}^{L}\beta_{i\ell,k}^2\frac{M}{K-1}}.
\end{aligned}
\label{proof_C1}
\end{equation}
In fact, a lower bound will approach the exact value when the number of antennas at the BS grows. From \eqref{proof_C1} with $\mu=M/K$, \eqref{eq:sinr1} is obtained.
\end{proof}
\lbn

\begin{figure*}[t]
\begin{align}\label{eq:exact}
R_{\ell,k} = \sum_{\mu=0}^{M-K}\frac{\log_{2}e}{(M-K-\mu)!} \times&\left[(-1)^{M-K-\mu-1}\left(\frac{1}{a^{M-K-\mu}}e^{\frac{1}{a}}\mathrm{Ei}\left(-\frac{1}{a}\right)-\frac{1}{b^{M-K-\mu}}e^{\frac{1}{b}}\mathrm{Ei}\left(-\frac{1}{b}\right)\right) \nonumber\right. \\
\hspace{0.7cm}
&+\left.\sum_{k=1}^{M-K-\mu}(k-1)!\left(\left(-\frac{1}{a}\right)^{M-K-\mu-k}-\left(-\frac{1}{b}\right)^{M-K-\mu-k}\right)\right].
\end{align}
\hrule
\end{figure*}

Unlike the limit SINR in \eqref{13}, the limit SINR in \eqref{eq:sinr1} depends on the transmit power of the user. In practice, however, the effect of the transmit power will vanish quickly in a high power regime. In addition, the SINR in \eqref{eq:sinr1} decreases significantly when $K \rightarrow M$ and is equivalent to \eqref{13} as $K \ll M$. This shows that as $K$ increases, the impact of the pilot contamination will become more severe.

\lbn
\begin{corollary}
Let $E_{\mathrm{u}}\triangleq P_{\mathrm{u}}M$, where $E_{\mathrm{u}}$ and $K$ are fixed. When $M$ grows without bound, the effective uplink $\mathrm{SINR}$ for the $k$-th user in the $\ell$-th cell of the ZF receiver can be deterministically approximated to
\begin{equation}
\begin{aligned}
\overline{\mathrm{SINR}}_{\ell,k}=\frac{\tau_{\mathrm{u}}E_{\mathrm{u}}^2\beta_{\ell\ell,k}^2}{\tau_{\mathrm{u}}E_{\mathrm{u}}^2\sum_{i=1,i\neq\ell}^{L}\beta_{i\ell,k}^2+M}.
\label{eq:sinr2}
\end{aligned}
\end{equation}
\end{corollary}
\lbn

\begin{proof}
In a manner similar to \eqref{14}, from \eqref{1} and \eqref{8}, we have
\begin{small}
\begin{equation}
\begin{aligned}
&\mathbf{r}_{\ell} = \mathbf{A}_{\ell}^{\dagger}\mathbf{y}_{\ell} = 
\left( \hat{\mathbf{G}}_{\ell\ell}^{\dagger} \hat{\mathbf{G}}_{\ell\ell}\right)^{-1} 
 \hat{\mathbf{G}}_{\ell\ell}^{\dagger} \Bigl(\sqrt{E_{\mathrm{u}}/M}\sum_{j=1}^{L}\mathbf{G}_{j\ell}\mathbf{x}_{j} + \mathbf{n}_{\ell}\Bigr) \\
&=\left( \hat{\mathbf{G}}_{\ell\ell}^{\dagger} \hat{\mathbf{G}}_{\ell\ell}\right)^{-1} 
\hat{\mathbf{D}}_{\ell\ell} \Bigl(\sqrt{E_{\mathrm{u}}M}\frac{\sum_{i=1}^{L}\sum_{j=1}^{L}\mathbf{G}_{i\ell}^{\dagger}\mathbf{G}_{j\ell} \mathbf{x}_{j}}{M}\\
&+ \sum_{i=1}^{L}\mathbf{G}_{i\ell}^{\dagger}\mathbf{n}_{\ell}  + \frac{1}{\sqrt{\tau_\mathrm{u}}}\sum_{j=1}^{L}\mathbf{W}_{\ell}^{\dagger}\mathbf{G}_{j\ell}\mathbf{x}_{j} +\frac{M}{\sqrt{\tau_{\mathrm{u}}E_{\mathrm{u}}}}\frac{\mathbf{W}_{\ell}^{\dagger} \mathbf{n}_{\ell}}{\sqrt{M}}\Bigr).
\end{aligned}
\label{eq:asymtotic1}
\end{equation}
\end{small}
When $K$ is fixed and $M$ goes to infinity, we have \cite{HC:70:Book}
\begin{equation}
\begin{aligned}
&\sum_{i=1}^{L}\sum_{j=1}^{L}\frac{\mathbf{G}_{i\ell}^{\dagger}\mathbf{G}_{j\ell}}{M}\stackrel{M\rightarrow\infty}{\rightarrow}\mathbf{D}_{\ell\ell}+\sum_{j=1,j\neq \ell}^{L}\mathbf{D}_{j\ell},\quad \forall i=j,\\
&\sum_{i=1}^{L}\frac{\mathbf{G}_{i\ell}^{\dagger}\mathbf{n}_{\ell}}{M}\stackrel{M\rightarrow\infty}{\rightarrow}\mathbf{0}\quad \mbox{and}\quad \sum_{j=1}^{L}\frac{\mathbf{W}_{\ell}^{\dagger}\mathbf{G}_{j\ell}}{M}\stackrel{M\rightarrow\infty}{\rightarrow}\mathbf{0},
\end{aligned}
\label{eq:asymtotic2}
\end{equation}
and the Lindeberg-Levy central limit theorem\footnote{With $\mathbf{h}_{i\ell,k}$ and $\mathbf{h}_{j\ell,k}$ defined in \eqref{limit}, from Lindeberg-Levy central limit theorem, we have $\frac{1}{\sqrt{M}}\mathbf{h}_{i\ell,k}^{\dagger}\mathbf{h}_{j\ell,k}\stackrel{d.}{\rightarrow}\mathcal{CN}(0,1)$ as $M\rightarrow\infty$, where $\stackrel{d.}{\rightarrow}$ denotes convergence in distribution.} can be used to obtain
\begin{equation}
\begin{aligned}
\frac{\mathbf{W}_{\ell}^{\dagger} \mathbf{n}_{\ell}}{\sqrt{M}}\sim \mathcal{CN} \left(\mathbf{0},\mathbf{I}_M\right).
\end{aligned}
\label{eq:asymtotic3}
\end{equation}
Applying \eqref{eq:asymtotic2} and \eqref{eq:asymtotic3} to \eqref{eq:asymtotic1}, we obtain \eqref{eq:sinr2}.
\end{proof}
\lbn

When $E_{\mathrm{u}}$ is fixed, we can also reduce the transmit power of each user in proportion to $1/M$ as $M$ grows, while preserving the transmission rate \cite{NMDL:13:VT}. However, as seen in \eqref{eq:sinr2}, interferences from other cells do not vanish. In particular, we can see that the interference due to pilot reuses, also known as pilot contamination, increases as $M$ grows, since the term $M$ is in the denominator of \eqref{eq:sinr2}.


\section{Achievable Rate, Outage Probability, and Symbol Error Rate}
\label{Performance}

In this section, we derive three important performance measures for the system: achievable uplink rate, outage probability, and SER.
\subsection{Achievable Uplink Rate}

\begin{proposition}
\label{Proposition3}
The exact achievable uplink rate for the $k$-th user in the $\ell$-th cell is given by \eqref{eq:exact}, which is shown at the top of this page, where $\mathrm{Ei}(z) \triangleq -\int_{-z}^{\infty}\frac{e^{-t}}{t}dt$,  $a \triangleq \frac{\sum_{i=1}^{L}\beta_{i\ell,k}^{2}}{\hat{\beta}_{\ell k}\left(\sum_{i=1}^{L}\alpha_{i\ell}+1/P_{\mathrm{u}}\right)}$, and $b \triangleq \frac{\sum_{i\neq \ell}^{L}\beta_{i\ell,k}^{2}}{\hat{\beta}_{\ell k}\left(\sum_{i=1}^{L}\alpha_{i\ell}+1/P_{\mathrm{u}}\right)}$.
\end{proposition}

	
\lbn	
\begin{proof}
\begin{equation}
\begin{aligned}
 &R_{\ell,k} = \mathbb{E}_{X}\left\{\log_{2}\left(1+\mathbf{\gamma}_{\ell,k}\right)\right\} \\
	&=\int_{0}^{\infty}\log_{2}\left(1+\mathbf{\gamma}_{\ell,k}(x)\right)p_{X}(x)dx\\
	&=\frac{\hat{\beta}_{\ell k}}{(M-K)!\beta_{\ell\ell,k}^{2}}\times\left\{\int_{0}^{\infty}\log_{2}\left(1+a.\hat{x}\right)\hat{x}^{M-K}e^{-\hat{x}}dx \right.\\
	&\left.-\int_{0}^{\infty}\log_{2}\left(1+b.\hat{x}\right)\hat{x}^{M-K}e^{-\hat{x}}dx\right\} \\
	&=\frac{\log_{2}e}{(M-K)!}\times\left\{\int_{0}^{\infty}\ln\left(1+a.\hat{x}\right)\hat{x}^{M-K}e^{-\hat{x}}d\hat{x}\right. \\
	&\left.-\int_{0}^{\infty}\ln\left(1+b.\hat{x}\right)\hat{x}^{M-K}e^{-\hat{x}}d\hat{x}\right\},
\end{aligned}
\label{29}
\end{equation} 
where $\hat{x}=\frac{\hat{\beta}_{\ell k}}{\beta_{\ell\ell,k}^{2}}x$. Using~\cite[Eq.~(4.337.5)]{GR:07:Book}, \eqref{eq:exact} can be derived. 
\end{proof}

%
 %
\subsection{Outage Probability}

The outage probability $P_{\mathrm{out}}$ is defined as the probability for which the instantaneous SINR falls below a given threshold, $\mathbf{\gamma}_{\mathrm{th}}$. It is easy to compute $P_{\mathrm{out}}$ from \eqref{25} and \eqref{26} as
\begin{equation}
\textit{P}_{\mathrm{out}} = \left\{ \begin{array}{rcl}
 1-e^{-\frac{\theta \mathbf{\gamma}_{\mathrm{th}}}{1-\eta \mathbf{\gamma}_{\mathrm{th}}}}\sum_{i=0}^{M-K}\frac{1}{i!}\left(\frac{\theta \mathbf{\gamma}_{\mathrm{th}}}{1-\eta \mathbf{\gamma}_{\mathrm{th}}}\right)^{i}, 
& \mathbf{\gamma}_{\mathrm{th}}<1/\eta \\ 
1, & \mathbf{\gamma}_{\mathrm{th}}\geq1/\eta
\end{array}\right.
\label{32}
\end{equation} 
\noindent where $\theta$ and $\eta$ are defined as in Proposition 1.

\subsection{Uplink SER}

From the SINR distribution in \eqref{22}, we can evaluate the SER of the system. We see that it is difficult to evaluate the SER directly from the PDF. Therefore, we take the moment generating function (MGF) approach \cite{ SA:00:Book} to derive a closed-form expression for the SER. We consider only the $\mathcal{M}$-QAM ($\mathcal{M}$ = $2^{i}$ with $i$ even).

\lbn
\begin{proposition}
\label{Proposition4}
The average SER for the transmission from the $k$-th user in $\ell$-th cell to the BS is given by

\begin{equation}
\begin{aligned}
\mathrm{SER}_{\ell,k}=&\frac{4}{\pi}\left(1-\frac{1}{\sqrt{\mathcal{M}}}\right)\left[\int_{0}^{\pi/2}\Phi_{\mathbf{\gamma}_{\ell,k}}\left(\frac{\mathsf{g}_{\mathsf{\mathcal{M}QAM}}}{\sin^{2}\phi}\right)d\phi \right.\\
&-\left.\left(1-\frac{1}{\sqrt{\mathcal{M}}}\right)\int_{0}^{\pi/4}\Phi_{\mathbf{\gamma}_{\ell,k}}\left(\frac{\mathsf{g}_{\mathsf{\mathcal{M}QAM}}}{\sin^{2}\phi}\right)d\phi\right],
\end{aligned}
\label{eq:ser}
\end{equation}
\noindent where $\mathsf{g}_{\mathsf{\mathcal{M}QAM}}$$\triangleq$ $\frac{3}{2(\mathcal{M}-1)}$ and $\Phi_{\mathbf{\gamma}_{\ell,k}}$(s) denotes the MGF of $\mathbf{\gamma}_{\ell,k}$, defined as
\begin{equation}
\begin{aligned}
\Phi_{\mathbf{\gamma}_{\ell,k}}(s)&=\mathbb{E}_{\mathbf{\gamma}_{\ell,k}}\left\{e^{-\mathbf{\gamma}_{\ell,k}s}\right\}\\
&=\sum_{p=0}^{M-K+1}\left(_{\quad p}^{M-K+1}\right)\left(\frac{-\beta_{\ell\ell,k}^{2}s}{\beta_{\ell\ell,k}^{2}s+\hat\beta_{\ell k}\theta}\right)^{p}\\
&\times\tensor[_2]{F}{_0}\left(M-K+1,p;-;\frac{-\kappa\hat\beta_{\ell k}}{\hat\beta_{\ell k}\theta+\beta_{\ell\ell,k}^2}s\right), 
\end{aligned}
\label{eq:mgf}
\end{equation}
where $\kappa\triangleq\frac{\tau_{\mathrm{u}}P_{\mathrm{u}}\hat{\beta}_{\ell k}+1}{\tau_{\mathrm{u}}P_{\mathrm{u}}\hat{\beta}_{\ell k}}$, and $\tensor[_2]{F}{_0}\left(\cdot\right)$ represents the generalized hypergeometric function~\cite[Eq.~(9.14.1)]{GR:07:Book}.
\end{proposition}
\lbn

\lbn
\begin{proof}
Let $\mathbf{\gamma}_{\ell,k}\triangleq\frac{X}{Y}$, where $Y\triangleq\theta+\eta X$. From ~\cite[Theorem~(3.19)]{YG:05:Book} and ~\cite[Theorem~(3.21)]{YG:05:Book}, we obtain $p_{Y}(y)=\frac{e^{-(y-\theta)/\kappa}}{(M-K)!\kappa}\left(\frac{y-\theta}{\kappa}\right)^{M-K}$. Correspondingly, the MGF of $\mathbf{\gamma}_{\ell,k}$ is given by
\begin{equation}
\Phi_{\mathbf{\gamma}_{\ell,k}}(s)=\mathbb{E}_{\mathbf{\gamma}_{\ell,k}}\left\{e^{-\mathbf{\gamma}_{\ell,k}s}\right\}=\int_{0}^{\infty}\Phi_{X}(s)p_{Y}(y)dy,
\label{eq:mgf1}
\end{equation}
where
\begin{equation}
\begin{aligned}
\Phi_{X}(s)&=\mathbb{E}_{X}\left\{e^{-\mathbf{\gamma}_{\ell,k} s}\right\}=\int_{0}^{\infty}e^{-\mathbf{\gamma}_{\ell,k} s}p_{X}(x)dx\\
&=\left(\frac{\hat\beta_{\ell k}y}{\hat\beta_{\ell k}y+\beta_{\ell\ell, k}^{2}s}\right)^{M-K+1}.
\label{eq:mgf2}
\end{aligned}
\end{equation}
Substituting \eqref{eq:mgf2} into \eqref{eq:mgf1}, we have
\begin{small}
\begin{equation}
\begin{aligned}
&\Phi_{\mathbf{\gamma}_{\ell,k}}(s)=\\
&\int_{0}^{\infty}\left(\frac{\hat\beta_{\ell k}y}{\hat\beta_{\ell k}y+\beta_{\ell\ell, k}^{2}s}\right)^{M-K+1}\frac{e^{-(y-\theta)/\kappa}}{(M-K)!\kappa}\left(\frac{y-\theta}{\kappa}\right)^{M-K}dy.
\label{eq:mgf3}
\end{aligned}
\end{equation}
\end{small}
In \cite{NMDL:13:VT}, the authors calculated the MGF for the case with a perfect CSI. We can apply the same procedure as in~\cite[Eq.~(53)]{NMDL:13:VT} onto \eqref{eq:mgf3} to obtain \eqref{eq:mgf}.
\end{proof}

\begin{corollary}

An upper bound for the SER of the $k$-th user in the $\ell$-th cell of the uplink transmission, assuming ZF processing occurs in the BS, is given by

\begin{equation}
\begin{aligned}
\mathrm{SER}_{\ell,k}^{\mathrm{UB}}& = \left(\frac{5}{3\sqrt{\mathcal{M}}}-\frac{1}{\mathcal{M}}-\frac{2}{3}\right)\Phi_{\mathbf{\gamma}_{\ell,k}}\left(\mathsf{g}_{\mathsf{\mathcal{M}QAM}}\right)\\
&+\left(1-\frac{1}{\sqrt{\mathcal{M}}}\right)\Phi_{\mathbf{\gamma}_{\ell,k}}\left(\frac{4}{3}\mathsf{g}_{\mathsf{\mathcal{M}QAM}}\right)\\
&+\left(\frac{2}{\sqrt{\mathcal{M}}}-1-\frac{1}{\mathcal{M}}\right)\Phi_{\mathbf{\gamma}_{\ell,k}}\left(2\mathsf{g}_{\mathsf{\mathcal{M}QAM}}\right).
\label{eq:serup}
\end{aligned}
\end{equation}
\end{corollary}

\begin{proof}
We derive a new approximation to avoid finite integration over $\phi$. We start by using a similar methodology as in \cite{CDS:03:WCOM} given as
\begin{equation}
\begin{aligned}
&\mathrm{SER}_{\ell,k}=\\
&\frac{4}{\pi}\left(1-\frac{1}{\sqrt{\mathcal{M}}}\right)
\times\mathbb{E}_{\mathbf{\gamma}_{\ell,k}}\left\{\int_{0}^{\pi/2}\Phi_{\mathbf{\gamma}_{\ell,k}}\left(\frac{\mathsf{g}_{\mathsf{\mathcal{M}QAM}}\mathbf{\gamma}_{\ell,k}}{\sin^{2}\phi}\right)d\phi\right. \\
&\left.-\left(1-\frac{1}{\sqrt{\mathcal{M}}}\right)\int_{0}^{\pi/4}\Phi_{\mathbf{\gamma}_{\ell,k}}\left(\frac{\mathsf{g}_{\mathsf{\mathcal{M}QAM}}\mathbf{\gamma}_{\ell,k}}{\sin^{2}\phi}\right)d\phi\right\},
\label{eq:proofup}
\end{aligned}
\end{equation}
where

\begin{equation}
\begin{aligned}
&\frac{2}{\pi}\int_{0}^{\pi/2}\exp\left(\frac{\mathsf{g}_{\mathsf{\mathcal{M}QAM}}\mathbf{\gamma}_{\ell,k}}{\sin^{2}\phi}\right)d\phi\\
&\approx \frac{1}{6}e^{\mathsf{g}_{\mathsf{\mathcal{M}QAM}}\mathbf{\gamma}_{\ell,k}}+\frac{1}{2}e^{\frac{4\mathsf{g}_{\mathsf{\mathcal{M}QAM}}\mathbf{\gamma}_{\ell,k}}{3}}.
\label{eq:proofup1}
\end{aligned}
\end{equation}
\begin{equation}
\begin{aligned}
&\frac{2}{\pi}\int_{0}^{\pi/4}\exp\left(\frac{\mathsf{g}_{\mathsf{\mathcal{M}QAM}}\mathbf{\gamma}_{\ell,k}}{\sin^{2}\phi}\right)d\phi \\
&\leq \frac{1}{2}e^{\mathsf{g}_{\mathsf{\mathcal{M}QAM}}\mathbf{\gamma}_{\ell,k}}+\frac{1}{2}e^{2\mathsf{g}_{\mathsf{\mathcal{M}QAM}}\mathbf{\gamma}_{\ell,k}}.
\label{eq:proofup2}
\end{aligned}
\end{equation}
Substituting \eqref{eq:proofup1} and \eqref{eq:proofup2} into \eqref{eq:proofup}, we obtain \eqref{eq:serup}.
\end{proof}

\section{Numerical Results}
\label{Simulation}

In this section, we provide numerical results that can be used to verify the above analysis. We first consider a simple scenario where the large-scale fading (path loss and shadow fading) is fixed and $L = 7$ cells share the same time-frequency. We set the number of users and the number of pilot symbols to $K = \tau_{\mathrm{u}} = 10$, and we set the length of the coherence interval to $T = 196$ \cite{MARZ:10:WCOM,HLM:13:COM}. We also assume that all of the desired links are associated with the gain $\beta_{\ell\ell,k} = 1$, and that all the interfering links for the intercell interference are associated with $\beta_{i\ell,k} \triangleq \beta$, $\forall i \neq \ell, k = 1,2,\cdots, K$, which is referred to as the cross gain factor. We define SNR $\triangleq P_{\mathrm{u}}$ since the noise variance is normalized to unity, and  the average transmit power of the data and pilot signals are assumed to be equal. Furthermore, we define the spectral efficiency per cell in bits/s/Hz for the $\ell$-th cell as
\begin{equation}
\mathrm{S} \triangleq \left(\frac{T-\tau_{\mathrm{u}}}{T}\right)\sum_{k=1}^{K} R_{\ell,k},
\label{33}
\end{equation} 
where $R_{\ell,k}$ is given in \eqref{eq:exact}.
\begin{figure}[t]
\centering
\includegraphics[width=0.45\textwidth]{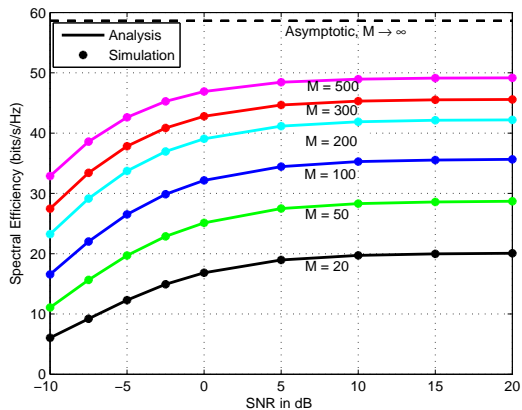}
\caption{Spectral efficiency versus the SNR for $M = 20, 50, 100, 200, 300, 500, \mbox{and } \infty$ ($L = 7$, $K = 10$, $\beta = 0.05$).}
\label{sumrateM}
\end{figure} 

\begin{figure}[t]
\centering
\includegraphics[width=0.5\textwidth]{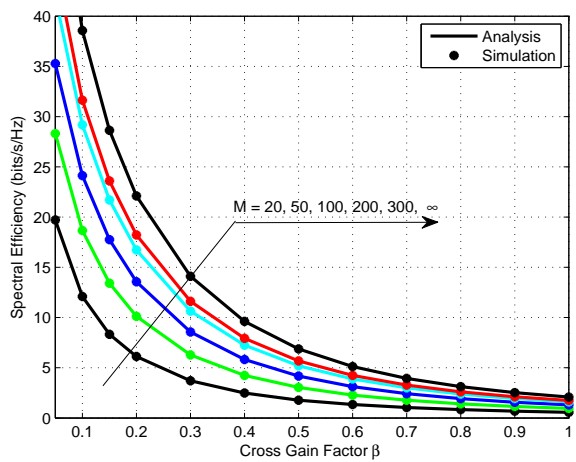}
\caption{Spectral efficiency versus the cross gain factor $\beta$ for $M = 20, 50, 100, 200, 300, \mbox{and } \infty$ ($L = 7$, $K = 10$, SNR = 10dB).}
\label{sumratealpha}
\end{figure}

\begin{figure}[t]
\centering
\includegraphics[width=0.5\textwidth]{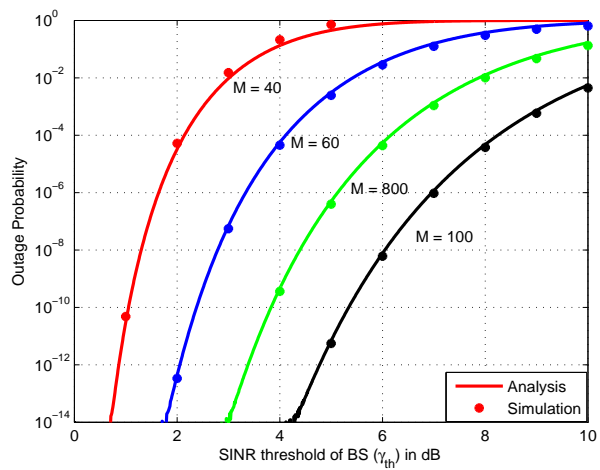}
\caption{Outage probability $P_{\mathrm{out}}$ versus the SINR threshold $\mathbf{\gamma}_{\mathrm{th}}$ for $M = 40, 60, 80, \mbox{and } 100$ ($L = 7$, $K = 10$, $\beta = 0.05$, SNR = 10dB).}
\label{outage2}
\end{figure}

\begin{figure}[t]
\centering
\includegraphics[width=0.5\textwidth]{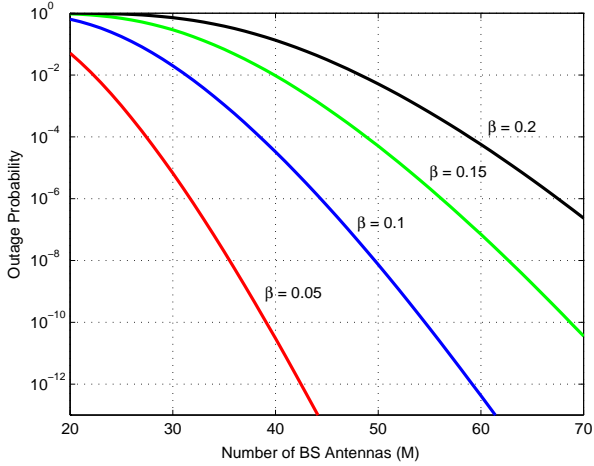}
\caption{Outage probability $P_{\mathrm{out}}$ versus the number of antennas $M$ for $\beta = 0.05, 0.1, 0.15, \mbox{and } 0.2$ ($L = 7$, $K = 10$, SNR = 10dB, $\mathbf{\gamma}_{\mathrm{th}} = 1\mbox{dB}$).}
\label{outage_M}
\end{figure}

\begin{figure}[t]
\centering
\includegraphics[width=0.5\textwidth]{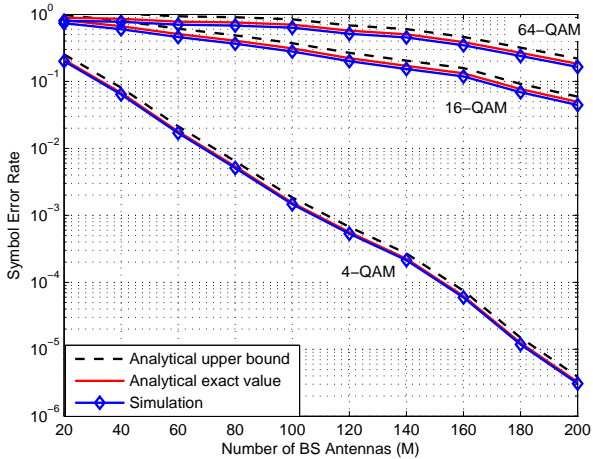}
\caption{SER versus the number of antennas $M$ for $\mathcal{M} = 4, 16, \mbox{and } 64$ ($L = 7$, $K = 10$, SNR = 10dB, $\beta = 0.1$).}
\label{SER}
\end{figure}
			
Fig. \ref{sumrateM} shows the spectral efficiency per cell versus the SNR for the cross gain factor $\beta = 0.05$ and various values of $M$. The simulation results were obtained through Monte-Carlo simulations using \eqref{15},\footnote{Note that $\alpha_{i\ell}$ and $\hat{\mathbf{G}}_{\ell\ell}$ in \eqref{15} are estimated by $\alpha_{i\ell}=\mathrm{tr}\bigl(\mathbb{E}[\xi_{i\ell}\xi_{i\ell}^{\dag}]\bigr)$ and  \eqref{eq:Estimation_H}, respectively.} and the analytical results are computed from \eqref{eq:exact}. As was expected, the spectral efficiency increases as $M$ increases. At a low SNR, the spectral efficiency increases steeply with an increase in SNR; at a high SNR, the spectral efficiency reaches a saturated value. This implies that increasing the transmit power for each user does not improve the system performance, but we need to increase the number of BS antennas instead. When $M$ grows without bound, we have an asymptotic curve resulting from the imperfect CSI due to the pilot contamination. We can observe how the spectral efficiency for $M \rightarrow \infty$ is independent of SNR, which validates our analysis in \eqref{13}. Finally, the simulation results and the analytical results are found to be in perfect agreement.
     
The effect of the cross gain factor for different values of $M$ is depicted in Fig. \ref{sumratealpha}. We also see how the spectral efficiency increases as $M$ increases. Note that each element of the channel estimation error $\mathbf{\xi}_{\ell\ell,k}$ follows a distribution with a zero mean and variance of $\beta_{\ell\ell,k}-\tau_{\mathrm{u}}P_{\mathrm{u}}\beta_{\ell\ell,k}^2\left(\tau_{\mathrm{u}}P_{\mathrm{u}}\sum_{j=1}^L\beta_{j\ell,k}+1\right)^{-1}$. When the cross gain factor increases from $0$ to $1$, or correspondingly the variance of the channel estimation error increases from $0.01$ to $0.85$, the spectral efficiency decreases significantly. In particular, when the cross gain factor is almost equal to the gain of the desired link, the spectral efficiency becomes very low regardless of the number of antennas $M$.
		
Fig. \ref{outage2} shows the outage probability per user versus the SINR threshold $\mathbf{\gamma}_{\mathrm{th}}$, when the cross gain factor $\beta = 0.05$ and  $M = 40, 60, 80, \mbox{and } 100$. The analytical results are computed using \eqref{32}. When $M$ increases, the outage probability decreases significantly.	Fig. \ref{outage_M} shows the outage per user versus $M$ for several values of $\beta$. We can see that the outage probability increases significantly as $\beta$ increases.
		
		%
		
Fig. \ref{SER} presents the SER against the number of BS antennas per user for different $\mathcal{M} = 4, 16, 64$. The exact analytical curves and the upper bound of the analytical curves are computed using Proposition 4 and Corollary 3, respectively. We can see that the analytical results are accurate for all cases. The upper bound of the SER is also quite tight, so we can use it to evaluate the performance of system for simple cases. In terms of the SER, this figure shows the advantages of using a massive number of antennas. The performance is seen to improve substantially as $M$ increases.
		
We now evaluate the effect of the imperfect CSI on the system performance in more practical scenarios. We consider hexagonal cells with a radius of $1,000$ m. The number of users and the number of pilot symbols are set to $K = \tau_{\mathrm{u}} = 10$. The large-scale fading coefficients are chosen as
\begin{equation}
\beta_{i\ell,k} = \frac{z_{i\ell,k}}{r_{i\ell,k}^{\gamma}},
\label{eq:large-scale}
\end{equation}
where $z_{i\ell,k}$ is a log-normal random variable with a standard deviation of 8$\mathrm{dB}$. $r_{i\ell,k}$ is the distance between the user and the BS, and $\gamma$ is the path loss exponent, which is assumed to be equal to $4$. The users are assumed to be located randomly and uniformly over the cell area. We consider the channel estimation error in cell 1 ($\ell=1$), and we use a normalized channel estimation error to evaluate the system performance for the MMSE estimation as
\begin{equation}
\mathsf{err} \triangleq  10\log_{10}\left(\frac{\left\|{\mathbf{G}}_{11}-\hat{\mathbf{G}}_{11}\right\|_F^2}{\left\|\mathbf{G}_{11}\right\|_F^2}\right).
\label{err}
\end{equation} 

\begin{figure}[t]
\centering
\includegraphics[width=0.5\textwidth]{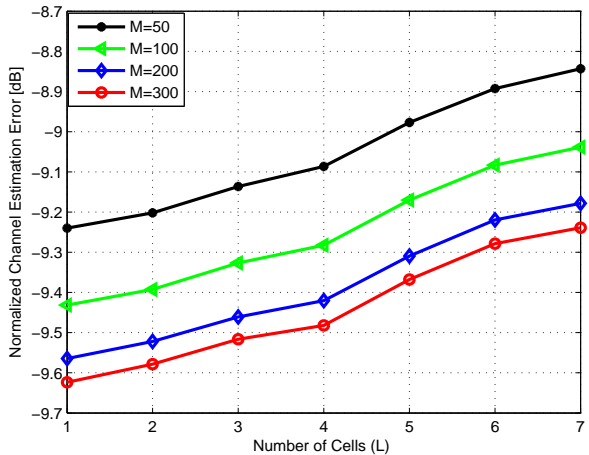}
\caption{Normalized channel estimation error versus the number of cells $L$ for $M = 50, 100, 200, \mbox{and } 300$ ($K = \tau_\mathrm{u} = 10$, SNR = 10dB).}
\label{nerr}
\end{figure}

\begin{figure}[t]
\centering
\includegraphics[width=0.5\textwidth]{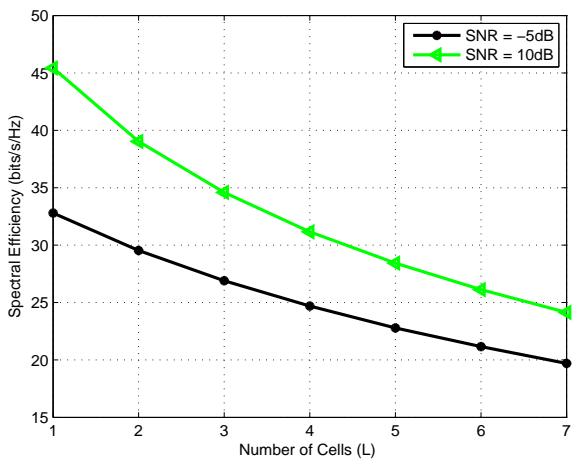}%
\caption{Spectral efficiency versus the number of cells $L$ for SNR = -5dB and 10dB ($K = \tau_\mathrm{u} = 10$, $M = 100$).}
\label{capacity_L}
\end{figure}

Fig. \ref{nerr} and \ref{capacity_L} illustrate the effect of having imperfect CSI against the number of cells $L$. For both cases, when the number of cells increases, the normalized channel estimation error in \eqref{err} increases and the spectral efficiency in \eqref{33} significantly decreases. In particular, when $L = 1$, which corresponds to an interference-free scenario, the performance for both cases is the best that can be obtained, but it degrades as $L$ increases, due to pilot contamination.

\section{Conclusion}
\label{Conclusion}

In this paper, we have analyzed the uplink performance of a multicell massive MIMO system where the BS is implemented with a ZF receiver with imperfect CSI due to pilot contamination. We have derived the exact analytical expressions for the PDF of the uplink SINR as well as the corresponding achievable rate, outage probability, and symbol error rate. Numerical results have been provided to validate the analysis, and the results indicate that in order to improve the system performance, we need to increase the number of BS antennas rather than the transmit power of the users. In addition, the effect of the pilot contamination on the system performance was investigated. The results presented in this paper have been provided as closed-form expressions, thus enabling easy and clear assessment of the system performance.

\end{document}